%% file: main.tex
\documentclass{article}

\usepackage{microtype}
\usepackage[utf8]{inputenc}

\usepackage{amsmath,amsthm,amsfonts}
\usepackage{enumitem}
\usepackage[margin=1in]{geometry}
\usepackage{url}

\newcommand{\BO}[1]{{ O}\left(#1\right)}

\newcommand{\BTO}[1]{\tilde{ O}\left(#1\right)}

\newcommand{\BT}[1]{{\Theta}\left(#1\right)}

\newcommand{\h}{\mathcal{H}}

\DeclareMathOperator*{\argmin}{arg\,min}
\DeclareMathOperator*{\E}{\mathbf{E}}

\newtheorem{theorem}{Theorem}

\newtheorem{definition}[theorem]{Definition}

\newcommand{\citep}{\cite}

\newcommand{\OUT}{\textnormal{OUT}}

\newcommand{\seclabel}[1]   {\label{sec:#1}}

\title{Adaptive MapReduce Similarity Joins \\
{\large Extended abstract}
}

\date{}

\author{Samuel McCauley \footnote{ {BARC and IT U. Copenhagen, Denmark.}
\protect\url{samc@itu.dk}}
\and
Francesco Silvestri
\footnote{University of Padova, Italy.
\protect\url{silvestri@dei.unipd.it}}
}

\begin{document} 

\maketitle

\input{abstract}

\input{intro}
\input{prelim}
\input{algo}
\input{concl}
\input{ack}

\bibliographystyle{plain}
\bibliography{all}

\end{document}

%% file: abstract.tex
\begin{abstract}

Similarity joins are a fundamental database operation.  Given data sets $S$ and $R$, the goal of a similarity join is to find all points $x \in S$ and $y \in R$ with distance at most $r$.  
  Recent research has investigated how locality-sensitive hashing (LSH) can be used for similarity join, and in particular two recent lines of work have made exciting progress on LSH-based join performance.  Hu, Tao, and Yi (PODS 17) investigated joins in a massively parallel setting, showing strong results that adapt to the size of the output.  Meanwhile, Ahle, Aum\"{u}ller, and Pagh (SODA 17) showed a sequential algorithm that adapts to the structure of the data, matching classic bounds in the worst case but improving them significantly on more structured data.

    We show that this adaptive strategy can be adapted to the parallel setting, combining the advantages of these approaches.  In particular, we show that a simple modification to Hu et al.'s algorithm achieves bounds that depend on the density of points in the dataset as well as the total outsize of the output.  Our algorithm uses no extra parameters over other LSH approaches (in particular, its execution does not depend on the structure of the dataset), and is likely to be efficient in practice.  
\end{abstract}

%% file: intro.tex
\section{Introduction}
\seclabel{intro}
 
        Similarity search is a fundamental problem in computer science where we seek to find items that are similar to one another.
        In this paper, we focus on the problem of {similarity joins} for high dimensional data, which can be viewed as a large number of batched similarity searches.  In particular, given two sets $R$ and $S$, we wish to find all pairs $(x,y)$ (with $x\in R$ and $y\in S$) where $x$ and $y$ have similarity above some threshold $r$.  
        Our results are largely agnostic to the particular similarity function used; for example one can immediately apply our techniques to Hamming distance, $\ell_1$ or $\ell_2$ distances, cosine similarity, Jaccard or Braun-Blanquet similarity \cite{ChristianiPaSi17} or even more exotic measures like Frechet distance \cite{DriemelS17}.  
Similarity joins have wide-ranging applications, such as web deduplication \cite{Bayardo_WWW07}, document clustering \cite{Broder_NETWORK97}, and data cleaning \cite{Arasu_VLDB06}.

        Unfortunately, similarity joins are extremely computationally intensive.\footnote{In fact, 
        this large computation may be unavoidable for some metrics, see e.g. \cite{Williams18}.}  For this reason, when performing similarity joins on large datasets it is often useful to use massively parallel machines using frameworks like MapReduce and Spark.
        A recent work \cite{HuTaYi17} has proposed an output sensitive MapReduce algorithm that leverages Locality Sensitive Hashing (LSH).
       When executed on $p$ machines and on two relations containing $n$ tuples, their solution requires $\BO{1}$ rounds and load  (i.e., maximum number of messages received/sent by a processor) 
       \[
    \label{eq:hu}
    \BO{\sqrt{\frac{\OUT_r }{p}p^{\rho/(1+\rho)}} + \sqrt{\frac{\OUT_{cr}}{p}}+\frac{n}{p}p^{\rho/(1+\rho)}},
  \]
where $\OUT_r$ is the  number of pairs with distances smaller than or equal to $r$,  $\OUT_{cr}$ is the number of pairs with distances in the range $(r,cr]$, and $\rho\in [0,1]$ is a value characterizing the LSH  (see Section \ref{sec:lsh}).
        This bound highlights some limitations of the standard LSH approach: there is an $\OUT_{cr}$ term due to ``false positives'' of the LSH, and there is a multiplicative term $p^{\rho/(1+\rho)}$ in the $\OUT_r$ contribution due to near pairs being reported multiple times.
     
     In this paper we show how the load of the previous algorithm can be improved by exploiting the novel LSH approach presented in \citep{AhleAuPa17}, which leverages a multi-level LSH data structure for solving the range reporting problem.
  Specifically, we set a small term $\kappa$ based on $p$ and the LSH parameters. For each $1\leq i < \kappa$, we say that a point in $x\in R$ is \emph{$i$-dense} if its number of near points in $S$ is in the range $[n p_2^{i-1},n p_2^{i})$; 
a point is \emph{$\kappa$-dense} if its number of near points in $S$ is smaller than $n p_2^{\kappa}$ (similar definitions hold for points in $S$).  
We let $\OUT_r,i$ denote the number of near pairs containing at least one $i$-dense point.
In this paper, we describe an MPC algorithm requiring $\BO{1}$ rounds and load
\begin{equation*}
\BTO{\sqrt{\left(\sum_{i=0}^\kappa \frac{\OUT_{r,i}}{p p_1^i }\right)} + \sqrt{\frac{\OUT_{cr}}{p}} + \frac{n}{p}p^{\rho/(1+\rho)}}.
\end{equation*}

        In the most extreme cases, there are no dense clusters (i.e. all close pairs contain only $\kappa$-dense points), in which case we get the same performance as \citep{HuTaYi17}. 
        However, the data structure has any dense clusters we get improved bounds. For example, if there is a cluster of $\Omega(n)$ close points, the first term of Hu et al. is $O(n/p^{1/(2(1+\rho))})$ while for us it is $O(n/p^{1/2})$.  Our bounds give a smooth decrease in performance as the size of the cluster decreases.

\subsection{Related Work}
Exact similarity search has been widely studied in the literature; we refer to \cite{Augsten13} for a survey. 
Approximate algorithms for similarity join often rely on LSH, with the underlying idea to adapt the indexing approach in \cite{Gionis_VLDB99}. 
This approach has been adapted for use in the I/O model \cite{PaghPhSi15} and in the MPC model \cite{HuTaYi17}.

A novel sequential LSH approach was recently introduced~\citep{AhleAuPa17} which dynamically adapts to the difficulty of each query; our result implements this idea in a massively parallel setting.  In short, this approach uses a simple recursive stopping rule to adapt to the structure of the dataset; we adapt this rule to the similarity join setting. 
Recently, the paper \cite{ChristianiPaSi17} adapted this approach for similarity join under Braun-Blanquet similarity; our results share the same basic principles but apply to more general distance metrics.


%% file: prelim.tex
\section{Preliminaries}
\seclabel{prelim}
In this section, we describe the adopted computational model and some relevant results on equi-joins, similarity search and LSH.
\subsection{Computational model}
In the literature, there are several computational models for massively parallel systems aiming  at describing MapReduce-like systems (e.g., \cite{PietracaprinaPuRi12,Beame13,Karloff10}). 
The majority of these models are very close to the bulk-synchronous parallel (BSP) model by L. Valiant \cite{Valiant90}, and in general they differ from the BSP on the cost functions (e.g., round number vs running time), parameters (e.g., local and global memory vs bandwidth and latency), and on some modeling aspects (e.g., a dynamic number of processing units vs a fixed number of processors in order to capture elastic settings in cloud).

In this paper,  in continuity with the previous work on similarity join in MapReduce \cite{HuTaYi17}, we use the \emph{Massively Parallel Computational} (MPC) model in \cite{Beame13}. 
It consists of $p$ processors $P_1, P_2,\ldots , P_p$ that are connected by a complete network.
In each round, each server does some local computation  and then sends messages to other servers, which will be received at the beginning of the subsequent round.
The complexity of the algorithm is the \emph{number of rounds} and the \emph{load} $L$, which is the maximum size of sent/received messages by each processor in each round.
For simplicity, we assume in the paper that $p<n^\epsilon$ for some constant $0<\epsilon<1$: this implies that sorting and prefix sum computations on input size $n$ can be performed in $\BO{1}$ rounds and load $\BO{n/p}$.
In general, the goal is to design MapReduce algorithms with a constant number of rounds. However, some works have shown that there are some inherent tradeoffs between round number and total communciation cost \cite{CeccarelloS15,afrati_et_al}.

\subsection{Equi-Join}
Let $R, S$ be two relations of total size $n$. The \emph{equi-join} of $R$ and $S$, denoted with $R \bowtie S$, is the set containing all pairs $(r,s)$ such that $r\in R$, $s \in S$ and $r=s$.
Hu et al. \cite{HuTaYi17} provided an optimal output sensitive and $\BO{1}$-round MPC algorithm for equi-join, with the following bounds:
\begin{theorem}[\cite{HuTaYi17}]\label{thm:equijoin}
There is an optimal deterministic algorithm that computes the equi-join between two relations of total size $n$ in $\BO{1}$ rounds and with load $\BT{\sqrt{O/p} + n/p}$, where $O=|R \bowtie S|$ denotes the equi-join size. It does not assume any prior statistical information about the data.
\end{theorem}
The term optimal holds for tuple-based algorithms, that is algorithms where tuples are atomic elements that must be processed or communicated in their entirety (i.e., indivisibility assumption).

\subsection{Similarity search problems}
Consider  a space $\mathbb{U}$ and a distance function $d: \mathbb{U}\times \mathbb{U} \rightarrow {\bf R}$. Let $r>0$ be an input parameter.
The \emph{$r$-near neigbor} problem is defined as follows: given a set $R$  of $n$ points from $\mathbb{U}$ and a query point $q\in \mathbb{U}$, return a point $x$ in $R$ at distance $d(x,q)\leq r$ if it exists.
The approximate version of the problem, named the \emph{$(r,c)$-near neighbor} problem where $c>1$ is the approximation factor, returns a point $x$ in $R$ with distance $d(x,q)\leq cr$ when there exists a point $x'\in R$ with distance $d(x',q)\leq r$.
The \emph{$r$-range reporting problem} requires to find all points at distance at most $r$ from a given query  $q$.

For convenience, we say that a pair $(x,y)$ is \emph{far} if $d(x,y)> cr$ (those that should not be reported), \emph{near} if 
$d(x,y)\leq r$  (those that
should be reported), and $c$-near if $r<d(x,y) \leq cr$
(those that should not be reported but the LSH provides no collision guarantees).
We also assume that each point in $\mathbb{U}$ can be stored in $\BO{1}$ memory words and that $d(x,y)$ can be computed in constant time (it is easy to extend bounds to the general case).

The similarity join problem is a batch version of the near neigbor problem. 
Specifically, the \emph{similarity join\/} with radius $r>0$ on the sets $R, S \subseteq \mathbb{U}$ is defined as the set 
$R \bowtie_{\leq r} S = \{ (x, y)\in R\times S \; | \; d(x, y	)\leq r \}$.
We let $n=|R|+|S|$.
In this paper, we will give a randomized solution that finds all pairs in $R \bowtie_{\leq r} S$ with high probability.
The proposed solution will generate also pairs  at distance larger than $r$ (i.e., false positives): while false positives can be easily removed by checking the true distance before emitting a pair, they still affect the performance of the algorithm (see the term $\OUT_r$ in the load upper bound).

\subsection{Locality-Sensitive Hashing}\label{sec:lsh}
  Much of recent work on similarity search and join has focused on {locality-sensitive hashing} (LSH).  In LSH, we hash each item in the dataset to a single {hash bucket} given by a randomized hash function.  Once the entire dataset is hashed, we perform a brute-force comparison between all pairs of points in the bucket, returning any similar points found.
The key idea behind LSH is the following property.  At a high level, an LSH family must map similar points to the same bucket with a much higher probability than far points. Formally, we have: 
\begin{definition}\label{def:LSH}
Fix a distance function $d: \mathbb{U}\times \mathbb{U} \rightarrow {\bf R}$.
A \emph{locality-sensitive hash} (LSH) family $\h$ is a family of functions $h: \mathbb{U} \rightarrow {\bf R}$ such that for each pair $x,y \in X$ and a random $h\in \h$, for arbitrary $q\in \mathbb{U}$, whenever $d(q,x)\leq d(q,y)$ we have $\Pr[h(q)=h(x)]\geq \Pr[h(q)=h(y)]$.
\end{definition}
When an LSH is applied to solve the $(c,r)$-near neighbor problem, it is common to describe the scheme with the probabilities $p_1$ and $p_2$ defined as follows: for each $x,y$ with $d(x,y)\leq r$ then $\Pr[h(x)=h(y)]\geq p_1$; for each $x,y$ with $d(x,y)> cr$ then $\Pr[h(x)=h(y)]\leq p_2$ (there are no requirements for pairs with distance in $(r,cr]$).
We define $\rho(r_1,r_2)={\log p(r_1)}/{\log p(r_2)}$, where $p(r)$ is the collision probability at distance $r$, and define $\rho=\rho(r,cr)$.

        Oftentimes, probabilities $p_1$ and $p_2$ are constants.  Thus, hashing a single time would lead to poor recall (we would be likely to miss close points) and poor precision (most points---$np_2$ of them to be precise---would hash to a given bucket in expectation, so searching within a bucket would be extremely costly).  For many use cases, this can be handled using a two-pronged approach: we  {concatenate} many hash functions to improve precision, and use many independent  {repetitions} to improve recall.  This approach can be formalized as follows, using a {$k$-concatenated hash} $\h_k$.  
    Let $\mathcal{H}$ be a hash family with $p_1,p_2 = \Omega(1)$.  
    Let $H_k = (h_1,h_2,\ldots, h_k)$ be a hash function consisting of the concatenation of $k$ independent hash functions from $\mathcal{H}$, and let $k = \log_{1/p_2} n$. Then
\begin{itemize}[topsep=0pt]
    \item for $x$, $y$ with $s(x,y)\geq r$, $\Pr(H_k(x) = H_k(y)) \geq 1/n^{\rho}$.  Thus, 
        after $n^{\rho}$ independent repetitions, 
        $x$ and $y$ will hash to the same bucket with constant probability.  
    \item for $x$, $y$ with $s(x,y)\leq cr$, $\Pr(H_k(x) = H_k(y)) \leq 1/n$.  Thus, each bucket will contain one far point in expectation. 
\end{itemize}

        In a single-processor setting, this framework allows us to perform single similarity {searches} in $\tilde{O}(n^{\rho})$ time.  
    Let $R$ and $S$ be sets of size $n$. Let $R\bowtie_r S$ denote the set of pairs $x\in R$ and $y\in S$ with $s(x,y) \geq r$; likewise let $R\bowtie_{cr} S$ denote the set of pairs with $s(x,y)\geq cr$.  
    Then the join between $S$ and $R$ can be computed using a k-concatenated LSH with parameters $p_1$ and $p_2$ in time
    $
        O\left(n^{1 + \rho} + n^{\rho}|R\bowtie_{r} S| + |R\bowtie_{cr} S|\right).
    $
These LSH-based approaches form the basic building blocks of this paper.  

\subsection{Adaptive Near Neighbor}
Our work leverages the recent work by Ahle et al. \cite{AhleAuPa17} that presents a data structure for the range reporting problem (in which we wish to find all near points to a query). 
Let $N_r(q)$ be the number of near points to a query $q$ (similarly, $N_{cr}(q)$ is the number of $c$-near points). 
Let the \emph{expansion} $c^*_{q}$ be the largest value such that there are at most twice as many $c^*_{q}$-near points of $q$ than there are $r$-near points of $q$.
(i.e.  $N_{cr}(q)\leq 2 N_r(q)$).

\begin{theorem}[\cite{AhleAuPa17}] 
    Consider a query point $q$, and parameters $r>0$ and $c > 1$. 
Then, the near neighbors of $q$ can be reported with constant probability in time
\begin{itemize}
    \item $\BO{N_r(q) (n/N_r(q))^{\rho(r,c^*_q)}}$ if $c_q^*\geq c$, or
    \item $\BO{N_r(q) (n/N_r(q))^{\rho(r,c)}+N_{cr}(q)}$ if $c_q^*< c$.
\end{itemize}
\end{theorem}

Let $\h_k$ denote the LSH obtained concatenating $k$ randomly and uniformly selected hashes from $\h$ and let $K=\lceil \rho \log_{1/p_2} n\rceil$. 
The data structure leverages a multi-level LSH: 
in each level $0\leq k <K$, the input set is partitioned according to $\BO{p_1^{-k} \log k}$ hash  functions in $\h_k$.
For a given query $q$, we let $W_{q,k}$ be the cost for finding the near neighbor of $q$ using the LSH at level $k$.
Since $\Pr[h_k(q)=h_k(x)]=\Pr[h(q)=h(x)]^k$, the expected value of $W_{q,k}$ is:
\begin{align*}
\E[W_{q,k}] &= 
p_1^{-k} \left(1+\sum_{x\in R} \Pr[h(q)=h(x)]^k\right) .
\end{align*}
The value $k_x = \argmin_{k\in [K]} W_{q,k}$ gives the best level to use for finding all near neighbors of $q$.
The  data structure computes an estimate $\tilde W_{q,k}$  of $W_{q,k}$ by summing the sizes of buckets where $q$ collides at level $k$ and then computing $\tilde k_x = \argmin_{k\in [K]} \tilde W_{q,k}$.
It is relevant to recall that the cost of removing $x$ at level $k'_x$ is upper bounded by the cost at level $\lceil \log(n/N(q,r))/\log (1/p_2)\rceil$ (without knowing the actual output size).

%% file: algo.tex
\section{A constant-round algorithm for similarity join}
Let $\h$ be an $(r,cr,p_1,p_2)$ LSH family, let $\rho = \log p_1 / \log p_2$, and let $\h_k$ denote the LSH family obtained by concatenating $k\geq 1$ copies of independent and	identically distributed LSHs in $\h$.

At the high level, the algorithm constructs a multi-level LSH data structure as in \citep{AhleAuPa17}, where the  $i$-level uses $i$-concatenated LSHs and
the bottom levels (i.e., below a given threshold $\kappa$) are removed.
Then, the algorithm removes each input point $x \in R \cup S$ by searching in the buckets of LSHs at level $k_x$, which is a suitable value that reduces the communication cost for removing point $x$. 
We initially assume that all the $k_x$ values are known, and we will later show that they can be computed (with a slight increase in the communication complexity) in one round.

Let $\kappa =\lceil (\rho/(1+\rho))\log_{p_1^{-1}}p \rceil$. For each point $x\in R$, we define  $k_x = \argmin_{i\in[\kappa]} \E[W_{x,i}]$, that is:
\begin{align*}
k_x 
  =\argmin_{i\in[\kappa]}  \left(p_1^{-i}\left(1+\sum_{y\in S} Pr\left[h_i(x)=h_i(y) \right]\right)\right). 
\end{align*}
 We recall that the value $k_x$ defined in \cite{AhleAuPa17} is defined using $\kappa = \lceil \rho \log_{1/p_1} n \rceil$.  The term $k_y$ for each $y\in S$ is defined equivalently.

The algorithm consists of $\kappa$ phases, each one requiring $\BO{1}$ rounds. 
These phases can be executed sequentially, resulting in an $\BO{\kappa}$-round algorithm; however, since there is no dependency among phases, they can be executed concurrently to give an $\BO{1}$-round algorithm. 
We assume the input to be evenly distributed among the $p$ processors.
An input point $x\in R \cup S$  is said \emph{active} during the $k_x$-th phase, \emph{passive} during the $i$-th phase for each $1\leq i<k_x$, and \emph{dead} during the $i$-th phase for each $i>k_x$.
The $i$-th phase, with $i\in [\kappa]$, is organized as follows:
\begin{enumerate}
    \item Choose $t_i = \Theta(p_1^{-i})$ hash functions $h^i_1,\ldots, h^i_{t_i}$ in $\h_k$ randomly and independently, and broadcast them to all $p$ processors.
\item Verify if $x$, for each $x\in R\cup S$, is passive, active, or dead  by comparing the index $i$ with the value $k_x$.
\item For each point $x\in R \cup S$ and hash functions $h^i_1,\ldots, h^i_{t_i}$ in $\h_k$, generate a tuple with key $(i,j,h^i_j(x))$ and value $x$ (including its status as active or passive).  
\item Let $Q_{(i,j,\ell)}$ be the set of tuples with key $(i,j,\ell)$. Remove all tuples associated with a key  $(i,j,\ell)$ that does not have an active point in $Q_{(i,j,\ell)}$.
This step can be executed with a sort and a prefix-sum computation.
\item Using the equi-join algorithm from \cite{HuTaYi17} (see Theorem \ref{thm:equijoin}) on the remaining tuples,  generate all pairs $(x,y)$ such that $x\in R$, $y\in S$ and at least one of $x$ and $y$ is active, and then output only near pairs (i.e., $d(x,y)\leq r$). 
\end{enumerate}

When the $\kappa$ phases are run concurrently, each entry $x$ is at the same time passive, active and dead. 
In other words, for each entry $x$ we generate the hash values $(i,j,h^i_j(x)$ for each $1\leq i \leq k_x$ and $1\leq j < t_i$. 
Note that the $\kappa$  equi-joins can be run as a single equi-join, with a slight improve of the load. We have the following theorem.

\begin{theorem}\label{thm:alg_const_round}
The above algorithm runs in  $\BO{1}$ rounds and  load:
\begin{equation*}
\BTO{\sqrt{\left(\sum_{i=0}^\kappa \frac{\OUT_{r,i}}{p p_1^i }\right)} + \sqrt{\frac{\OUT_{cr}}{p}} + \frac{n}{p^{1/(1+\rho)}}}.
\end{equation*}
\end{theorem}
\begin{proof}
The correctness of the algorithm follows from \cite{AhleAuPa17}.
We now upper bound the cost of Step 4.
    Consider a point $x$ in $R\cup S$, and let $x$ be $\ell_x$-dense. Denote with $N_r(x)$ and $N_{cr}(x)$  the number of near and $cr$-near points to $x$ respectively. We observe that
     \[
        \OUT_{r,i} = \Theta\left(\sum_{\text{$i$-dense x}} N_r(x)\right), \text{ and } \OUT_{cr} = \Theta\left(\sum_x N_{cr}(x) \right). 
    \]
    By construction, the total number of output pairs involving $x$ is upper bounded by $W_{k_x,x} \leq W_{\ell_x,x}$.  We now bound $W_{\ell_x, x}$; summing over these gives the total number of output pairs.  
    
    If $\ell_x = \kappa$, then $W_{\ell_x,x} \leq N_r(x)/p_1^{\kappa} + N_{cr}(x) + n(p_2/p_1)^{\kappa}$. 
    We have $n(p_2/p_1)^{\kappa} = n/p_1^{\kappa (1-1/\rho)} = \Theta(n/pp_1^{2\kappa})$.
    
    If $\ell_x < \kappa$, then $W_{\ell_x,x} \leq N_r(x)/p_1^{\ell_x} + N_{cr}(x) + n(p_2/p_1)^{\ell_x}$. Since $x$ is $\ell_x$-dense, 
    $N_r(x)/p_1^{\ell_x} = \Theta(n(p_2/p_1)^{\ell_x})$. 
    Thus, $W_{\ell_x,x} = O(N_r(x)/p_1^{\ell_x} + N_{cr}(x))$.
  
    Summing over all $x$ and invoking Theorem~\ref{thm:equijoin}, Step 4 has load 
    \[
        \BTO{\sqrt{\left(\sum_{i=0}^\kappa \frac{\OUT_{r,i}}{p_1^i}+ {\OUT_{cr}} + \frac{n^2}{pp_1^{2\kappa}}\right){\frac{1}{p}}} }
    \]
    which reduces to the claimed bounds using Jensen's inequality.

We now consider the first three steps. 
We observe that the number of processors is at most  $n^\epsilon$ for some constant $\epsilon>0$, and hence prefix-sum, broadcast, sorting require constant rounds and load $\BO{I/p}$, where $I$ is the input size.
Therefore the first three steps require $\BO{1}$ rounds.
Step 1 requires load $\BTO{p_1^{-\kappa}/p}$.
    For Step 2, we observe that each input point is only copied once for each hash function; since there are  $\BTO{p_1^{-\kappa}}$ hash functions and each processor contains $\BO{n/p}$ points, the load is $\BTO{n p_1^{-\kappa}/p}=\frac{n}{p^{1/(1+\rho)}}$ (indeed the sorting step guarantee load balancing).
Step 3 consists of a sorting and prefix sum on $\BTO{n p_1^{-\kappa}}$ points and hence its load is $\BTO{n p_1^{-\kappa}/p}=\frac{n}{p^{1/(1+\rho)}}$.
Therefore, the first three steps meet the claimed bound and the theorem follows.
\end{proof}

The previous algorithm assumes that $k_x$ is known for each input point. 
However, it is easy to see that the values can be computed in $\BO{1}$ rounds and load 
$\BTO{\frac{n}{p^{1/(1+\rho)}}}$, which is negligible compared to the overall cost of the algorithm. 
Indeed, it suffices to generate all keys up to level $\kappa$,  sort them and compute some prefix sums to estimate the expected costs $\E[W_{i,x}]$ as in \citep{AhleAuPa17}. 

%% file: concl.tex
\section{Conclusion}
We have seen how to improve the output sensitivity of the algorithm in \cite{HuTaYi17} with a simple $\BO{1}$-round algorithm.
A limit of our approach is the $n/p^{1/(1+\rho)}$ term in the load, which  is due to the computation of $k_x$ values (and to a non tight bound in the proof of Theorem \ref{thm:alg_const_round}).
We conjecture that it is possible to reduce this term for some output densities by using approximations to the $k_x$ values.
One promising approach is the technique for the Braun-Blanquet similarity in \cite{ChristianiPaSi17}, where $k_x$ values are replaced with greedily-computed collision probabilities.
However, such an approach seems to increase the number of rounds to $\BO{\log n}$.

Another interesting direction is an experimental evaluation of our approach. 
Analysis for specific inputs in \cite{AhleAuPa17} and experiments in \cite{ChristianiPaSi17} indicate that our recursive approach may give speedups beyond the worst-case theoretical analysis. It would be interesting to see the load incurred using our approach on practical datasets.

%% file: ack.tex
\section{Acknowledgements}
We would like to thank Rasmus Pagh, Johan Sivertsen, Shikha Singh, John Augustine, and Mohit Daga for helpful discussions. 
We also thank the participants of the AlgoPARC Workshop on Parallel Algorithms and Data Structures, in part supported by the NSF grant no. 1745331.
The authors were supported in part by the  ERC grant agreement no.~614331, and by project SID2017 of the University of Padova.
  BARC, Basic Algorithms Research Copenhagen, is supported by VILLUM Foundation grant 16582.